\documentclass{article}
\usepackage[accepted]{icml2020}

\usepackage{microtype}
\pdfoutput=1
\usepackage[hidelinks,breaklinks]{hyperref}
\usepackage{booktabs}
\usepackage{paralist}
\usepackage{marginnote}
\usepackage{color}
\usepackage{float}
\usepackage[T1]{fontenc}

\usepackage{amsmath}
\let\savebmatrix\bmatrix
\let\saveendbmatrix\endbmatrix
\renewenvironment{bmatrix}{\savebmatrix}{\saveendbmatrix}

\usepackage{mathtools}
\usepackage{amssymb}
\makeatletter
\providecommand*{\shuffle}{%
  \mathbin{\mathpalette\shuffle@{}}%
}
\newcommand*{\shuffle@}[2]{%
  \sbox0{$#1\vcenter{}$}%
  \kern .15\ht0 
  \rlap{\vrule height .25\ht0 depth 0pt width 2.5\ht0}%
  \raise.1\ht0\hbox to 2.5\ht0{%
    \vrule height 1.75\ht0 depth -.1\ht0 width .17\ht0 %
    \hfill
    \vrule height 1.75\ht0 depth -.1\ht0 width .17\ht0 %
    \hfill
    \vrule height 1.75\ht0 depth -.1\ht0 width .17\ht0 %
  }%
  \kern .15\ht0 
}
\makeatother

\usepackage{amsthm}
\newtheorem{theorem}{Theorem}
\newtheorem{proposition}[theorem]{Proposition}
\newtheorem{lemma}[theorem]{Lemma}
\theoremstyle{definition}
\newtheorem{defn}{Definition}
\newtheorem{example}{Example}

\newcommand{\fwd}{\mathrm{fw}}

\newcommand{\pe}{\mathbf{e}}

\newcommand{\transpose}{^\mathsf{T}}
\DeclareMathOperator{\diag}{diag}

\newcommand{\sym}[1]{\mathtt{#1}}

\usepackage{newtxtext,newtxmath} 
\usepackage{tikz}

\usetikzlibrary{automata,arrows}
\tikzset{state/.style={draw,circle,minimum width=6pt}}
\tikzset{every edge/.style={draw,->,>=stealth',shorten >=1pt,auto,semithick}}
\tikzset{initial text={},double distance=2pt}
\usetikzlibrary{calc}
\usetikzlibrary{decorations.pathreplacing}

\usepackage{pgfplots}
\pgfplotsset{trig format plots=rad}
\usepgfplotslibrary{groupplots}

\newlength\figureheight
\newlength\figurewidth
\usepackage{natbib}
\begin{document}
\twocolumn[
\icmltitle{Representing Unordered Data Using Complex-Weighted Multiset Automata}

\begin{icmlauthorlist}
\icmlauthor{Justin DeBenedetto}{nd}
\icmlauthor{David Chiang}{nd}
\end{icmlauthorlist}
\icmlaffiliation{nd}{Department of Computer Science and Engineering,
University of Notre Dame, Notre Dame, IN 46556, USA}
\icmlcorrespondingauthor{Justin DeBenedetto}{jdebened@nd.edu}
\icmlcorrespondingauthor{David Chiang}{dchiang@nd.edu}
\vskip 0.3in
]

\printAffiliationsAndNotice{}

\begin{abstract}
Unordered, variable-sized inputs arise in many settings across multiple fields. The ability for set- and multiset-oriented neural networks to handle this type of input has been the focus of much work in recent years. We propose to represent multisets using complex-weighted \emph{multiset automata} and show how the multiset representations of certain existing neural architectures can be viewed as special cases of ours. Namely, (1) we provide a new theoretical and intuitive justification for the Transformer model's representation of positions using sinusoidal functions, and (2) we extend the DeepSets model to use complex numbers, enabling it to outperform the existing model on an extension of one of their tasks.  
\end{abstract}

\section{Introduction}

Neural networks which operate on unordered, variable-sized input \citep{vinyals2015order,Wagstaff2019OnTL} have been gaining interest for various tasks, such as processing graph nodes \citep{Murphy2018JanossyPL}, hypergraphs \citep{maron2018invariant}, 3D~image reconstruction \citep{yang2019}, and point cloud classification and image tagging \citep{Zaheer2017DeepS}. Similar networks have been applied to multiple instance learning \citep{Pevn2016UsingNN}.  

One such model, DeepSets \citep{Zaheer2017DeepS}, computes a representation of each element of the set, then combines the representations using a commutative function (e.g., addition) to form a representation of the set that discards ordering information. \citet{Zaheer2017DeepS} provide a proof that any function on sets can be modeled in this way, by encoding sets as base-4 fractions and using the universal function approximation theorem, but their actual proposed model is far simpler than the model constructed by the theorem.

In this paper, we propose to compute representations of multisets using \emph{weighted multiset automata}, a variant of weighted finite-state (string) automata in which the order of the input symbols does not affect the output. This representation can be directly implemented inside a neural network. We show how to train these automata efficiently by approximating them with string automata whose weights form complex, diagonal matrices.

Our representation is a generalization of DeepSets', and it also turns out to be a generalization of the Transformer's position encodings \citep{Vaswani2017AttentionIA}. In Sections~\ref{sec:position} and \ref{sec:deepsets}, we discuss the application of our representation in both cases. 
\begin{itemize}
\item Transformers \citep{Vaswani2017AttentionIA} encode the position of a word as a vector of sinusoidal functions that turns out to be a special case of our representation of multisets. So weighted multiset automata provide a new theoretical and intuitive justification for sinusoidal position encodings. We experiment with several variations on position encodings inspired by this justification, and although they do not yield an improvement, we do find that learned position encodings in our representation do better than learning a different vector for each absolute position.

\item We extend the DeepSets model to use our representation, which amounts to upgrading it from real to complex numbers. On an extension of one of their tasks (adding a sequence of one-digit numbers and predicting the units digit), our model is able to reach perfect performance, whereas the original DeepSets model does no better than chance.
\end{itemize}

\section{Definitions}

We first review weighted string automata, then modify the definition to weighted multiset automata.

\subsection{String automata}

Weighted finite automata are commonly pictured as directed graphs whose edges are labeled by symbols and weights, but we use an alternative representation as collections of matrices: for each symbol $a$, this representation takes all the transitions labeled $a$ and arranges their weights into an adjacency matrix, called $\mu(a)$. This makes it easier to use the notation and techniques of linear algebra, and it also makes clearer how to implement weighted automata inside neural networks.

Let $\mathbb{K}$ be a commutative semiring; in this paper, $\mathbb{K}$ is either $\mathbb{R}$ or $\mathbb{C}$. 

\begin{defn}
A \emph{$\mathbb{K}$-weighted finite automaton} (WFA) over $\Sigma$ is a tuple $M=(Q, \Sigma, \lambda, \mu, \rho)$, where:
\begin{itemize}
\item $Q = \{1, \ldots, d\}$ is a finite set of states,
\item $\Sigma$ is a finite alphabet,
\item $\lambda \in \mathbb{K}^{1 \times d}$ is a row vector of initial weights,
\item $\mu : \Sigma \rightarrow \mathbb{K}^{d \times d}$ assigns a transition matrix to every symbol, and
\item $\rho \in \mathbb{K}^{d \times 1}$ is a column vector of final weights.
\end{itemize}
\end{defn}
The weight $\lambda_q$ is the weight of starting in state $q$; the weight $[\mu(a)]_{qr}$ is the weight of transitioning from state $q$ to state $r$ on input $a$, and the weight $\rho_q$ is the weight of accepting in state $q$. We extend the mapping $\mu$ to strings: If $w = w_1 \cdots w_n \in \Sigma^\ast$, then $\mu(w) = \mu(w_1) \cdots \mu(w_n)$. Then, the weight of $w$ is $\lambda \mu(w) \rho$. In this paper, we are more interested in representing $w$ as a vector rather than a single weight, so define the vector of \emph{forward weights} of $w$ to be $ \fwd_M(w) = \lambda \mu(w)$. (We include the final weights $\rho$ in our examples, but we do not actually use them in any of our experiments.)
 
Note that, different from many definitions of weighted automata, this definition does not allow $\epsilon$-transitions, and there may be more than one initial state. (Hereafter, we use $\epsilon$ to stand for a small real number.)

\begin{example} \label{eg:wfa}
Below, $M_1$ is a WFA that accepts strings where the number of $\sym{a}$'s is a multiple of three; $(\lambda_1, \mu_1, \rho_1)$ is its matrix representation:
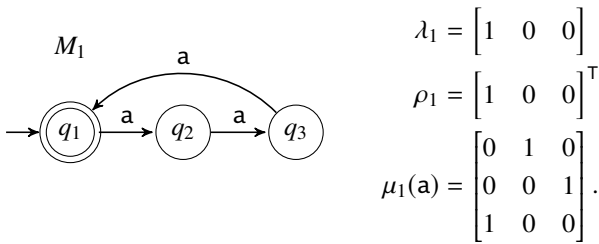
\[
\begin{tikzpicture}[baseline=0,x=1.5cm,y=1.5cm]
\node[initial,accepting,state](q1) {$q_1$};
\node[state](q2) at (1,0) {$q_2$};
\node[state](q3) at (2,0) {$q_3$};
\draw (q1) edge node {$\sym{a}$} (q2);
\draw (q2) edge node {$\sym{a}$} (q3);
\draw (q3) edge[bend right=45,auto=right] node {$\sym{a}$} (q1);
\node at (0,0.75) {$M_1$};
\end{tikzpicture}
\qquad
\begin{aligned}
\lambda_1 &= \begin{bmatrix} 1 & 0 & 0 \end{bmatrix} \\
\rho_1 &= \begin{bmatrix} 1 & 0 & 0 \end{bmatrix}\transpose \\
\mu_1(\sym{a}) &= \begin{bmatrix}
0 & 1 & 0 \\
0 & 0 & 1 \\
1 & 0 & 0 \\
\end{bmatrix}.
\end{aligned}
\]
And $M_2$ is a WFA that accepts $\{\sym{b}\}$; $(\lambda_2, \mu_2, \rho_2)$ is its matrix representation:
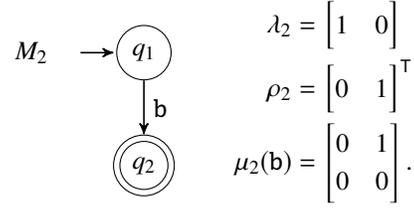
\[
\begin{tikzpicture}[baseline=-0.75cm,x=1.5cm,y=1.5cm]
\node[initial,state](q1) {$q_1$};
\node[accepting,state](q2) at (0,-1) {$q_2$};
\draw (q1) edge node {$\sym{b}$} (q2);
\node at (-1,0) {$M_2$};
\end{tikzpicture}
\qquad
\begin{aligned}
\lambda_2 &= \begin{bmatrix} 1 & 0 \end{bmatrix} \\
\rho_2 &= \begin{bmatrix} 0 & 1 \end{bmatrix}\transpose \\
\mu_2(\sym{b}) &= \begin{bmatrix}
0 & 1 \\
0 & 0 \\
\end{bmatrix}.
\end{aligned}
\]
\end{example}

\subsection{Multiset automata}

The analogue of finite automata for multisets is the special case of the above definition where multiplication of the transition matrices $\mu(a)$ does not depend on their order.
\begin{defn}
\label{def:multisetautomaton}
A \emph{$\mathbb{K}$-weighted multiset finite automaton} is one whose transition matrices commute pairwise. That is, for all $a, b \in \Sigma$, we have $\mu(a)\mu(b) = \mu(b)\mu(a)$.
\end{defn}

\begin{example} \label{eg:msa}
Automata $M_1$ and $M_2$ are multiset automata because they each have only one transition matrix, which trivially commutes with itself. Below is a multiset automaton that accepts multisets over $\{\sym{a}, \sym{b}\}$ where the number of $\sym{a}$'s is a multiple of three and the number of $\sym{b}$'s is exactly one.
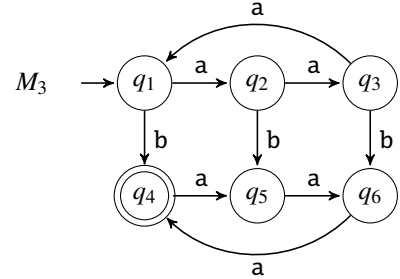
\begin{center}
\begin{tikzpicture}[x=1.5cm,y=1.5cm]
\node at (-1,0) {$M_3$};
\node[initial,state](q1) {$q_1$};
\node[state](q2) at (1,0) {$q_2$};
\node[state](q3) at (2,0) {$q_3$};
\node[accepting,state](q4) at (0,-1) {$q_4$};
\node[state](q5) at (1,-1) {$q_5$};
\node[state](q6) at (2,-1) {$q_6$};
\draw (q1) edge node {$\sym{b}$} (q4);
\draw (q2) edge node {$\sym{b}$} (q5);
\draw (q3) edge node {$\sym{b}$} (q6);
\draw (q1) edge node {$\sym{a}$} (q2);
\draw (q2) edge node {$\sym{a}$} (q3);
\draw (q3) edge[bend right=45,auto=right] node {$\sym{a}$} (q1);
\draw (q4) edge node {$\sym{a}$} (q5);
\draw (q5) edge node {$\sym{a}$} (q6);
\draw (q6) edge[bend left=45] node {$\sym{a}$} (q4);
\end{tikzpicture}
\\
\resizebox{\columnwidth}{!}{$\begin{aligned}
\lambda_3 &= \begin{bmatrix} 1 & 0 & 0 & 0 & 0 & 0 \end{bmatrix} &
\rho_3 &= \begin{bmatrix} 0 & 0 & 0 & 1 & 0 & 0 \end{bmatrix}\transpose \\
\mu_3(\sym{a}) &= \begin{bmatrix}
0 & 1 & 0 & 0 & 0 & 0 \\
0 & 0 & 1 & 0 & 0 & 0 \\
1 & 0 & 0 & 0 & 0 & 0 \\
0 & 0 & 0 & 0 & 1 & 0 \\
0 & 0 & 0 & 0 & 0 & 1 \\
0 & 0 & 0 & 1 & 0 & 0 
\end{bmatrix}
&
\mu_3(\sym{b}) &= \begin{bmatrix}
0 & 0 & 0 & 1 & 0 & 0 \\
0 & 0 & 0 & 0 & 1 & 0 \\
0 & 0 & 0 & 0 & 0 & 1 \\
0 & 0 & 0 & 0 & 0 & 0 \\
0 & 0 & 0 & 0 & 0 & 0 \\
0 & 0 & 0 & 0 & 0 & 0 
\end{bmatrix}.
\end{aligned}$}
\end{center}
\end{example}

Our multiset automata (Definition~\ref{def:multisetautomaton}) describe the \emph{recognizable} distributions over multisets \citep[Section 4.1]{sakarovitch:2009}. A recognizable distribution is intuitively one that we can compute by processing one element at a time using only a finite amount of memory.

This is one of two natural ways to generalize string automata to multiset automata. The other way, which describes all the \emph{rational} distributions of multisets \citep[Section 3.1]{sakarovitch:2009}, places no restriction on transition matrices and computes the weight of a multiset $w$ by summing over all paths labeled by permutations of $w$. But this is NP-hard to compute.

Our proposal, then, is to represent a multiset $w$ by the vector of forward weights, $\fwd_M(w)$, with respect to some weighted multiset automaton $M$. In the context of a neural network, the transition weights $\mu(a)$ can be computed by any function as long as it does not depend on the ordering of symbols, and the forward weights can be used by the network in any way whatsoever.

\section{Training}

Definition~\ref{def:multisetautomaton} does not lend itself well to training, because parameter optimization needs to be done subject to the commutativity constraint. Previous work \citep{debenedetto+chiang:2018} 
suggested approximating training of a multiset automaton by training a string automaton while using a regularizer to encourage the weight matrices to be close to commuting. However, this strategy cannot make them commute exactly, and the regularizer, which has $O(|\Sigma|^2)$ terms, is expensive to compute.

Here, we pursue a different strategy, which is to restrict the transition matrices $\mu(a)$ to be diagonal. This guarantees that they commute. As a bonus, diagonal matrices are computionally less expensive than full matrices. Furthermore, we show that if we allow complex weights, we can learn multiset automata with diagonal matrices which represent multisets almost as well as full matrices. We show this first for the special case of unary automata (\S\ref{sec:unary}) and then general multiset automata (\S\ref{sec:general}).

\subsection{Unary automata}
\label{sec:unary}

Call an automaton \emph{unary} if $|\Sigma|=1$. Then, for brevity, we simply write $\mu$ instead of $\mu(a)$ where $a$ is the only symbol in $\Sigma$. 

Let $\|\cdot\|$ be the Frobenius norm; by equivalence of norms \citep[p.~352]{horn+johnson:2012}, the results below should carry over to any other matrix norm, as long as it is monotone, that is: if $A \leq B$ elementwise, then $\|A\| \leq \|B\|$.

As stated above, our strategy for training a unary automaton is to allow $\mu$ to be complex, but restrict it to be diagonal. The restriction does not lose much generality, for suppose that we observe data generated by a multiset automaton $M = (\lambda, \mu, \rho)$. Then $\mu$ can be approximated by a complex diagonalizable matrix by the following well-known result:
\begin{proposition}[\protect{\citealt[p.~116]{horn+johnson:2012}}] \label{prop:dense}
For any complex square matrix $A$ and $\epsilon > 0$, there is a complex matrix $E$ such that $\|E\| \leq \epsilon$ and $A+E$ is diagonalizable in~$\mathbb{C}$. 
\end{proposition}
\begin{proof}
Form the Jordan decomposition $A = PJP^{-1}$. We can choose a diagonal matrix $D$ such that $\|D\| \leq \frac{\epsilon}{\kappa(P)}$ (where $\kappa(P) = \|P\| \|P^{-1}\|$) and the diagonal entries of $J+D$ are all different. Then $J+D$ is diagonalizable. Let $E = P D P^{-1}$; then $\|E\| \leq \|P\| \|D\| \|P^{-1}\| = \kappa(P) \|D\| \leq \epsilon$, and  $A+E = P(J+D)P^{-1}$ is also diagonalizable.
\end{proof}
So $M$ is close to an automaton $(\lambda, \mu+E, \rho)$ where $\mu+E$ is diagonalizable. Furthermore, by a change of basis, we can make $\mu+E$ diagonal without changing the automaton's behavior:
\begin{proposition}
If $M = (\lambda, \mu, \rho)$ is a multiset automaton and $P$ is an invertible matrix, the multiset automaton $(\lambda', \mu', \rho')$ where
\begin{align*}
\lambda' &= \lambda P^{-1} \\
\mu'(a) &= P \mu(a) P^{-1} \\
\rho' &= P \rho
\end{align*}
computes the same multiset weights as $M$.
\end{proposition}
This means that in training, we can directly learn complex initial weights $\lambda'$ and a complex diagonal transition matrix $\mu'$, and the resulting automaton $M'$ should be able to represent multisets almost as well as a general unary automaton would.

\begin{example}
In Example~\ref{eg:wfa}, $\mu_1(\sym{a})$ is diagonalizable:
\begin{align*}
\lambda_1' &= \begin{bmatrix} \frac13 & \frac13 & \frac13 \end{bmatrix} \\
\rho_1' &= \begin{bmatrix} 1 & 1 & 1 \end{bmatrix}\transpose \\
\mu_1'(\sym{a}) &= \diag \begin{bmatrix}
1 & \exp {\frac{2\pi}3i} & \exp {-\frac{2\pi}3i} 
\end{bmatrix}.
\end{align*}
On the other hand, $\mu_2(\sym{b})$ is only approximately diagonalizable:
\begin{align*}
\lambda_2' &= \begin{bmatrix} \frac1{2\epsilon} & -\frac1{2\epsilon} \end{bmatrix} \\
\rho_2' &= \begin{bmatrix} 1 & 1 \end{bmatrix}\transpose \\
\mu_2'(\sym{b}) &= \diag \begin{bmatrix}
\epsilon & -\epsilon
\end{bmatrix}.
\end{align*}
But $\mu_2'$ can be made arbitrarily close to $\mu_2$ by making $\epsilon$ sufficiently close to zero.
\end{example}

It might be thought that even if $\mu'$ approximates $\mu$ well, perhaps the forward weights, which involve possibly large powers of $\mu$, will not be approximated well. As some additional assurance, we have the following error bound on the powers of $\mu$:
\begin{theorem}
\label{prop:bound}
For any complex square matrix $A$, $\epsilon > 0$, and $0<r<1$, there is a complex matrix $E$ such that $A+E$ is diagonalizable in $\mathbb{C}$ and, for all $n \geq 0$,
\begin{align*}
\|(A+E)^n-A^n\| &\leq r^n \epsilon &&\text{if $A$ nilpotent,} \\
\frac{\|(A+E)^n-A^n\|}{\|A^n\|} &\leq n\epsilon &&\text{otherwise.}
\end{align*}
\end{theorem}
In other words, if $A$ is not nilpotent, the relative error in the powers of $A$ increase only linearly in the exponent; if $A$ is nilpotent, we can't measure relative error because it eventually becomes undefined, but the absolute error decreases exponentially.

\begin{proof}
See Appendix~\ref{sec:bound}.
\end{proof}

\subsection{General case}
\label{sec:general}

In this section, we allow $\Sigma$ to be of any size. Proposition~\ref{prop:dense} unfortunately does not hold in general for multiple matrices \citep{o2006approximately}. That is, it may not be possible to perturb a set of commuting matrices so that they are \emph{simultaneously} diagonalizable.  

\begin{defn}
Matrices $A_1,\ldots,A_m$ are \emph{simultaneously diagonalizable} if there exists an invertible matrix $P$ such that $PA_iP^{-1}$ is diagonal for all $i \in \{1,\cdots, n\}$.

We say that $A_1,\cdots,A_m$ are \emph{approximately simultaneously diagonalizable (ASD)} if, for any $\epsilon>0$, there are matrices $E_1, \ldots, E_m$ such that $\|E_i\| \leq \epsilon$ and $A_1+E_1, \ldots, A_m+E_m$ are simultaneously diagonalizable.
\end{defn}

\citet{o2006approximately} give examples of sets of matrices that are commuting but not ASD. However, if we are willing to add new states to the automaton (that is, to increase the dimensionality of the transition matrices), we can make them ASD.

\begin{theorem}
\label{prop:genASD}
For any weighted multiset automaton, there is an equivalent complex-weighted multiset automaton, possibly with more states, whose transition matrices are ASD.
\end{theorem}

We extend Proposition~\ref{prop:dense} from unary automata to non-unary multiset automata that have a special form; then, we show that any multiset automaton can be converted to one in this special form, but possibly with more states.

Let $\oplus$ stand for direct sum of vectors or matrices, $\otimes$ for the Kronecker product, and define the shuffle product (also known as the Kronecker sum) $A \shuffle B = A \otimes I + I \otimes B$. These operations extend naturally to weighted multiset automata \citep{debenedetto+chiang:2018}: If $M_A = (\lambda_A, \mu_A, \rho_A)$ and $M_B = (\lambda_B, \mu_B, \rho_B)$, then define
\begin{align*}
M_A \oplus M_B &= (\lambda_{\oplus}, \mu_{\oplus}, \rho_{\oplus}) &
\hspace{-1em} M_A \shuffle M_B &= (\lambda_{\shuffle}, \mu_{\shuffle}, \rho_{\shuffle}) \\
\lambda_{\oplus} &= \lambda_A \oplus \lambda_B &
\lambda_{\shuffle} &= \lambda_A \otimes \lambda_B \\
\mu_{\oplus}(a) &= \mu_A(a) \oplus \mu_B(a) &
\mu_{\shuffle}(a) &= \mu_A(a) \shuffle \mu_B(a) \\
\rho_{\oplus} &= \rho_A \oplus \rho_B &
\rho_{\shuffle} &= \rho_A \otimes \rho_B.
\end{align*}
$M_A \oplus M_B$ recognizes the union of the multisets recognized by $M_A$ and $M_B$; if $M_A$ and $M_B$ use disjoint alphabets, then $M_A \shuffle M_B$ recognizes the concatenation of the multisets recognized by $M_A$ and $M_B$.

They are of interest here because they preserve the ASD property, so non-unary automata formed by applying direct sum and shuffle product to unary automata are guaranteed to have ASD transition matrices.

\begin{proposition}
\label{prop:regops}
If $M_1$ and $M_2$ are multiset automata with ASD transition matrices, then $M_1 \oplus M_2$ has ASD transition matrices, and $M_1 \shuffle M_2$ has ASD transition matrices.
\end{proposition}
\begin{proof}
See Appendix~\ref{sec:regops}.
\end{proof}

\begin{example}
Consider $M_1$ and $M_2$ from Example~\ref{eg:wfa}. If we assume that $\mu_1(\sym{b})$ and $\mu_2(\sym{a})$ are zero matrices, then the shuffle product of $M_2$ and $M_1$ is exactly $M_3$ from Example~\ref{eg:msa}. So the transition matrices of $M_3$ are ASD:
\begin{align*}
\lambda_3' &= \begin{bmatrix}
\frac1{6\epsilon} & \frac1{6\epsilon} & \frac1{6\epsilon} & -\frac1{6\epsilon} & -\frac1{6\epsilon} &-\frac1{6\epsilon} 
\end{bmatrix} \\
\mu_3'(\sym{a}) &= \diag
\begin{bmatrix}
1 & e^{\frac{2\pi}{3}i} & e^{-\frac{2\pi}{3}i} & 1 & e^{\frac{2\pi}{3}i} & e^{-\frac{2\pi}{3}i}
\end{bmatrix} \\
\mu_3'(\sym{b}) &= \diag
\begin{bmatrix}
\epsilon & \epsilon & \epsilon & -\epsilon & -\epsilon & -\epsilon
\end{bmatrix} \\
\rho_3' &= \begin{bmatrix}
1 & 1 & 1 & 1 & 1 & 1
\end{bmatrix}\transpose.
\end{align*}
\end{example}

The next two results give ways of expressing multiset automata as direct sums and/or shuffle products of smaller automata. Lemma~\ref{lem:asdblocks} expresses a multiset automaton as a direct sum of smaller automata, without increasing the number of states, but does not guarantee that the smaller automata are ASD. Proposition~\ref{prop:makeasd} expresses a multiset automaton as a direct sum of shuffle products of unary automata, but can increase the number of states.

\begin{lemma}[\citealt{o2006approximately}] \label{lem:asdblocks}
Suppose $A_1,\ldots,A_k$ are commuting $n\times n$ matrices over an algebraically closed field.  Then there exists an invertible matrix $C$ such that $C^{-1}A_1C,\ldots,C^{-1}A_kC$ are block diagonal matrices with matching block structures and each diagonal block has only a single eigenvalue (ignoring multiplicities).  That is, there is a partition $n=n_1+\cdots + n_r$ of $n$ such that 
\begin{align}
    C^{-1}A_iC = B_i = \begin{bmatrix}
    B_{i1}\\
    &B_{i2}\\
    &&\ddots\\
    &&&B_{ir}
    \end{bmatrix}.
    \label{eq:asdblocks}
\end{align}
where each $B_{ij}$ is an $n_j \times n_j$ matrix having only a single eigenvalue for $i=1,\ldots,k$ and $j=1,\ldots,r$.  
\end{lemma}

\begin{proposition} \label{prop:makeasd}
If $M$ is a weighted multiset automaton with~$d$ states over an alphabet with $m$ symbols, there exists an equivalent complex-weighted multiset automaton with $\binom{2m+d}{d-1}$ states whose transition matrices are ASD.
\end{proposition}

\begin{proof}
See Appendix~\ref{sec:makeasd}.
\end{proof}

\begin{proof}[Proof of Theorem~\ref{prop:genASD}]
By Lemma~\ref{lem:asdblocks}, we can put the transition matrices into the form~(\ref{eq:asdblocks}). By Proposition~\ref{prop:makeasd}, for each $j$, we can convert $B_{1j}, \ldots B_{kj}$ into ASD matrices $B_{1j}', \ldots, B_{kj}'$, and by Proposition~\ref{prop:regops}, their direct sum $B_{1j}' \oplus \cdots \oplus B_{kj}'$ is also ASD.
\end{proof}
This means that if we want to learn representations of multisets over a finite alphabet $\Sigma$, it suffices to constrain the transition matrices to be complex diagonal, possibly with more states. Unfortunately, the above construction increases the number of states by a lot. But this does not in any way prevent the use of our representation; we can choose however many states we want, and it's an empirical question whether the number of states is enough to learn good representations.

The following two sections look at two practical applications of our representation.

\section{Position Encodings}
\label{sec:position}

\begin{table*}
\caption{Machine translation experiments with various position encodings. Scores are in case-insensitive BLEU, a common machine translation metric.  The best score in each column is printed in boldface.}
\newcommand{\sig}{\rlap{$^\dag$}}
\begin{center}
\begin{tabular}{@{}lllllllll|l@{}}
\toprule
& & \multicolumn{8}{c}{case-insensitive BLEU} \\
Model & Training & En-Vi$^\ast$ & Uz-En & Ha-En & Hu-En & Ur-En & Ta-En & Tu-En & combined \\
\midrule
diagonal polar & fixed & 32.6 & 25.7 & 24.4 & \textbf{34.2} & 11.5 & 13.4 & 25.7 & 26.4 \\
& learned angles & \textbf{32.7} & 25.8 & 25.4\sig & 34.0 & 11.1\sig & 14.1\sig & 25.7 & \textbf{26.6} \\
\midrule
full matrix & random & 32.6 & \textbf{25.9} & \textbf{25.6}\sig & 34.1 & 11.1\sig & 12.6\sig & \textbf{26.1} & 26.5 \\
& learned & 32.5 & 24.5\sig & 23.6 & 33.5 & 11.4 & \textbf{14.5}\sig & 23.8\sig & 26.5 \\
\midrule
per position & random & 32.6 & 24.3\sig & 24.6 & 33.6\sig & 11.1 & 14.0\sig & 25.7 & 26.3 \\
& learned & 32.0\sig & 22.6\sig & 21.2\sig & 33.0\sig & \textbf{11.7} & 14.4\sig & 21.1\sig & 25.0\sig \\
\bottomrule
\end{tabular}
\end{center}
$^\ast$tokenized references \\
$^\dag$significantly different from first line ($p<0.05$, bootstrap resampling)
\label{tab:position}
\end{table*}

One of the distinguishing features of the Transformer network for machine translation \citep{Vaswani2017AttentionIA}, compared with older RNN-based models, 
is its curious-looking \emph{position encodings},
\begin{equation}
\begin{aligned}
\pe^p_{2j-1} &= \sin 10000^{-2(j-1)/d} (p-1) \\
\pe^p_{2j} &= \cos 10000^{-2(j-1)/d} (p-1) 
\end{aligned}
\label{eq:position}
\end{equation}
which map word positions $p$ (ranging from 1 to $n$, the sentence length) to points in the plane and are the model's sole source of information about word order.

In this section, we show how these position encodings can be interpreted as the forward weights of a weighted unary automaton. We also report on some experiments on some extensions of position encodings inspired by this interpretation.

\subsection{As a weighted unary automaton}

Consider a diagonal unary automaton $M$ in the following form:
{\setlength\arraycolsep{3pt}
\begin{align*}
\lambda &= \begin{bmatrix} s_1 \exp i\phi_1 & s_1 \exp -i\phi_1 & s_2 \exp i\phi_2 & s_2 \exp -i\phi_2 & \cdots \end{bmatrix} \\
\mu &= \begin{bmatrix}
r_1 \exp i\theta_1 & 0 & 0 & 0 & \cdots \\
0 & r_1 \exp -i\theta_1 & 0 & 0 & \cdots \\
0 & 0 & r_2 \exp i\theta_2 & 0 & \dots \\
0 & 0 & 0 & r_2 \exp -i\theta_2 & \cdots \\
\vdots & \vdots & \vdots & \vdots & \ddots \\
\end{bmatrix}.
\end{align*}}%
In order for a complex-weighted automaton to be equivalent to some real-weighted automaton, the entries must come in conjugate pairs like this, so this form is fully general.

By a change of basis, this becomes the following unary automaton $M'$ (this is sometimes called the real Jordan form):
{\setlength\arraycolsep{3pt}
\begin{equation}
\begin{aligned}
\lambda' &= \begin{bmatrix} s_1 \cos \phi_1 & s_1 \sin \phi_1 & s_2 \cos \phi_2 & s_2 \sin \phi_2 & \cdots \end{bmatrix} \\
\mu' &= \begin{bmatrix}
r_1 \cos \theta_1 & r_1 \sin \theta_1 & 0 & 0 & \cdots \\
-r_1 \sin \theta_1 & r_1 \cos \theta_1 & 0 & 0 & \cdots \\
0 & 0 & r_2 \cos \theta_2 & r_2 \sin \theta_2 & \dots \\
0 & 0 & -r_2 \sin \theta_2 & r_2 \cos \theta_2 & \cdots \\
\vdots & \vdots & \vdots & \vdots & \ddots \\
\end{bmatrix}
\end{aligned}
\label{eq:diagpolar}
\end{equation}}%
so that for any string prefix $u$ (making use of the angle sum identities):
\begin{align*}
\fwd_{M'}(u)^\top &= \begin{bmatrix} 
s_1 \, r_1^{|u|} \, \cos (\phi_1 + |u|\theta_1) \\
s_1 \, r_1^{|u|} \, \sin (\phi_1 + |u|\theta_1) \\ 
s_2 \, r_2^{|u|} \, \cos (\phi_2 + |u|\theta_2) \\ 
s_2 \, r_2^{|u|} \, \sin (\phi_2 + |u|\theta_2) \\ 
\vdots
\end{bmatrix}.
\end{align*}
If we let
\begin{align*}
s_i &= 1 &
\phi_i &= \frac{\pi}{2} \\
r_i &= 1 &
\theta_j &= -10000^{-2(j-1)/d}
\end{align*}
this becomes exactly equal to the position encodings defined in (\ref{eq:position}). Thus, the Transformer's position encodings can be reinterpreted as follows: it runs automaton $M'$ over the input string and uses the forward weights of $M'$ just before position $p$ to represent $p$. This encoding, together with the embedding of word $w_p$, is used as the input to the first self-attention layer.

\subsection{Experiments}

This reinterpretation suggests that we might be able to learn position encodings instead of fixing them heuristically. \citet{Vaswani2017AttentionIA}, following \citet{gehring+:2017} and followed by \citet{devlin-etal-2019-bert}, learn a different encoding for each position, but multiset automata provide parameterizations that use many fewer parameters and hopefully generalize better.

We carried out some experiments to test this hypothesis, using an open-source implementation of the Transformer, Witwicky.\footnote{\url{https://github.com/tnq177/witwicky}} The settings used were the default settings, except that we used 8k joint BPE operations and $d=512$ embedding dimensions. We tested the following variations on position encodings.
\begin{itemize}
\item Diagonal polar: multiset automaton as in eq.~(\ref{eq:diagpolar})
\begin{itemize}
\item fixed: The original sinusoidal encodings \citep{Vaswani2017AttentionIA}.
\item learned angles: Initialize the $\phi_i$ and $\theta_i$ to the original values, then optimize them.
\end{itemize}
\item Full matrix: multiset automaton with real weights
\begin{itemize}
\item random: Randomize initial weights so that their expected norm is the same as the original, and transition matrix using orthogonal initialization \citep{saxe+:2013}, and do not optimize them.
\item learned: Initialize $\lambda$ and $\mu$ as above, and then optimize them.
\end{itemize}
\item Per position: a real vector for each position
\begin{itemize}
\item random: Choose a random vector with fixed norm for each absolute position, and do not optimize them.
\item learned: Initialize per-position encodings as above, then optimize them \citep{gehring+:2017}.
\end{itemize}
\end{itemize}

Table~\ref{tab:position} compares these methods on seven low-resource language pairs (with numbers of training tokens ranging from 100k to 2.3M), with the final column computed by concatenating all seven test sets together. Although learning position encodings using multiset automata (``diagonal polar, learned angles'' and ``full matrix, learned'') does not do better than the original sinusoidal encodings (the 0.2 BLEU improvement is not statistically significant), they clearly do better than learning per-position encodings, supporting our view that multiset automata are the appropriate way to generalize sinusoidal encodings.

\section{Complex DeepSets}
\label{sec:deepsets}

In this section, we incorporate a weighted multiset automaton into the DeepSets \citep{Zaheer2017DeepS} model, extending it to use complex numbers. Our code is available online.\footnote{\url{https://github.com/jdebened/ComplexDeepSets}}

\subsection{Models}

The DeepSets model computes a vector representation for each input symbol and sums them to discard ordering information. We may think of the elementwise layers as computing the log-weights of a diagonal multiset automaton, and the summation layer as computing the forward log-weights of the multiset. (The logs are needed because DeepSets adds, whereas multiset automata multiply.) However, DeepSets uses only real weights, whereas our multiset automata use complex weights. Thus, DeepSets can be viewed as using a multiset representation which is a special case of ours.

We conduct experiments comparing the DeepSets model \citep{Zaheer2017DeepS}, a GRU model, an LSTM model, and our complex multiset model. The code and layer sizes for the three baselines come from the DeepSets paper.\footnote{\url{https://github.com/manzilzaheer/DeepSets/blob/master/DigitSum/text_sum.ipynb}} See Figure~\ref{fig:Architecture} for layer types and sizes for the three baseline models.

\begin{figure*}
    \centering
    \scalebox{0.8}{%
    \begin{tikzpicture}[x=2.5cm,y=1.5cm]
    \begin{scope}[every node/.style={align=center,anchor=east}]
    \node at (-0.5,5) { LSTM };
    \node at (-0.5,4) { GRU };
    \node at (-0.5,3) { DeepSets };
    \node at (-0.5,1) { Ours };
    \end{scope}
    \begin{scope}[every node/.style={draw,rectangle,align=center,minimum width=2.1cm,minimum height=1.3cm}]
    \node(lstm-input) at (0,5) {Input Layer\\$n$};
    \node(lstm-embed) at (1,5) {Embedding \\$n \times 100$};
    \node(lstm-lstm) at (2.5,5) {LSTM \\$50$};
    \node(lstm-out) at (4,5) {Dense \\$1$};
    \draw (lstm-input) edge (lstm-embed);
    \draw (lstm-embed) edge (lstm-lstm);
    \draw (lstm-lstm) edge (lstm-out);
    
    \node(gru-input) at (0,4) {Input Layer\\$n$};
    \node(gru-embed) at (1,4) {Embedding \\$n \times 100$};
    \node(gru-gru) at (2.5,4) {GRU \\$80$};
    \node(gru-out) at (4,4) {Dense \\$1$};
    \draw (gru-input) edge (gru-embed);
    \draw (gru-embed) edge (gru-gru);
    \draw (gru-gru) edge (gru-out);
    
    \node(ds-input) at (0,3) {Input Layer\\$n$};
    \node(ds-embed) at (1,3) {Embedding \\$n \times 100$};
    \node(ds-dense) at (2,3) {Dense \\$n \times 30$};
    \node(ds-sum) at (3,3) {Sum \\$30$};
    \node(ds-out) at (4,3) {Dense \\$1$};
    \draw (ds-input) edge (ds-embed);
    \draw (ds-embed) edge (ds-dense);
    \draw (ds-dense) edge (ds-sum);
    \draw (ds-sum) edge (ds-out);

    \node(msa-input) at (0,1) {Input Layer\\$n$};
    \node(msa-embed-r) at (1,2) {Embedding $r$ \\$n \times 50$};
    \node(msa-embed-a) at (1,1) {Embedding $a$ \\$n \times 50$};
    \node(msa-embed-b) at (1,0) {Embedding $b$ \\$n \times 50$};
    \node(msa-prod) at (3,1) {Complex \\ Product \\$150$};
    \node(msa-out) at (4,1) {Dense \\$1$};
    \draw (msa-input) edge (msa-embed-r) edge (msa-embed-a) edge (msa-embed-b);
    \draw (msa-embed-r) edge (msa-prod);
    \draw (msa-embed-a) edge (msa-prod);
    \draw (msa-embed-b) edge (msa-prod);
    \draw (msa-prod) edge (msa-out);
    \end{scope}
    \end{tikzpicture}%
    }
    \caption{Models compared in Section~\ref{sec:deepsets}.  Each cell indicates layer type and output dimension.
    }
    \label{fig:Architecture}
\end{figure*}
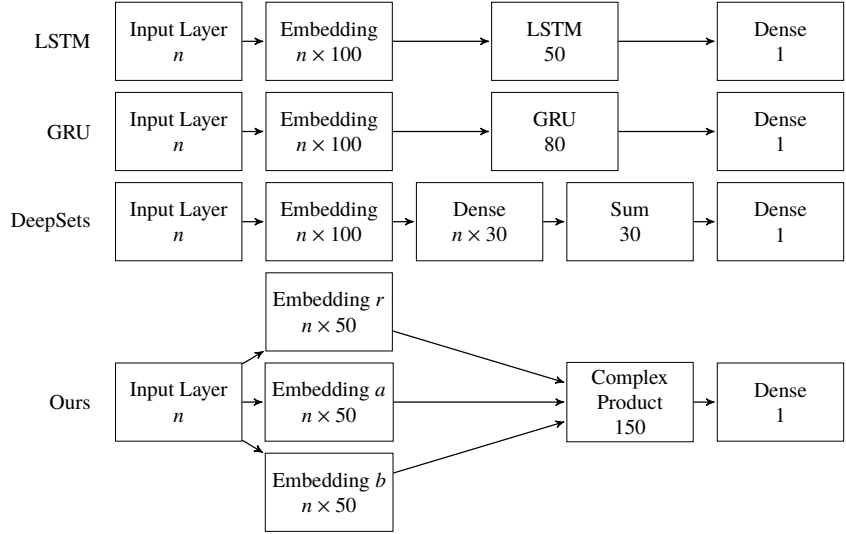

In our system, to avoid underflow when multiplying many complex numbers, we store each complex number as $e^r(a+bi)$ where $r$, $a$, and $b$ are real and $a$ and $b$ are normalized such that $a^2+b^2=1$ prior to multiplication. Thus, for each complex-valued parameter, we have three real-valued scalars ($r$, $a$, and $b$) to learn. To this end, each input is fed into three separate embedding layers of size $50$ (for $r$, $a$, and $b$).  Since the weight on states occurs in complex conjugate pairs within the diagonalized multiset automata, we only need to store half the states.  This is why we use $50$ rather than $100$ for our embeddings. (While the DeepSets code uses a dense layer at this point, in our network, we found that we could feed the embeddings directly into a complex multiplication layer to discard ordering information. This reduced the number of parameters for our model and did not affect performance.)  The output of this is then a new $r$, $a$, and $b$ which are concatenated and fed into a final dense layer as before to obtain the output.  Since our diagonalized automata have complex initial weights ($\lambda'$), we also tried learning a complex initial weight vector $\lambda'$, but this had no effect on performance.

The total number of parameters for each model was 4,161 parameters for the DeepSets model, 31,351 parameters for the LSTM model, 44,621 parameters for the GRU model, and 1,801 parameters for our model. In order to eliminate number of parameters as a difference from our model to the DeepSets model, we also tried the DeepSets model without the first dense layer and with embedding sizes of 150 to exactly match the number of parameters of our model, and the results on the test tasks were not significantly different from the baseline DeepSets model.

For tasks 1 and 2, we used mean squared error loss, a learning rate decay of 0.5 after the validation loss does not decrease for 2 epochs, and early stopping after the validation loss does not decrease for 10 epochs.  

\subsection{Experiments}

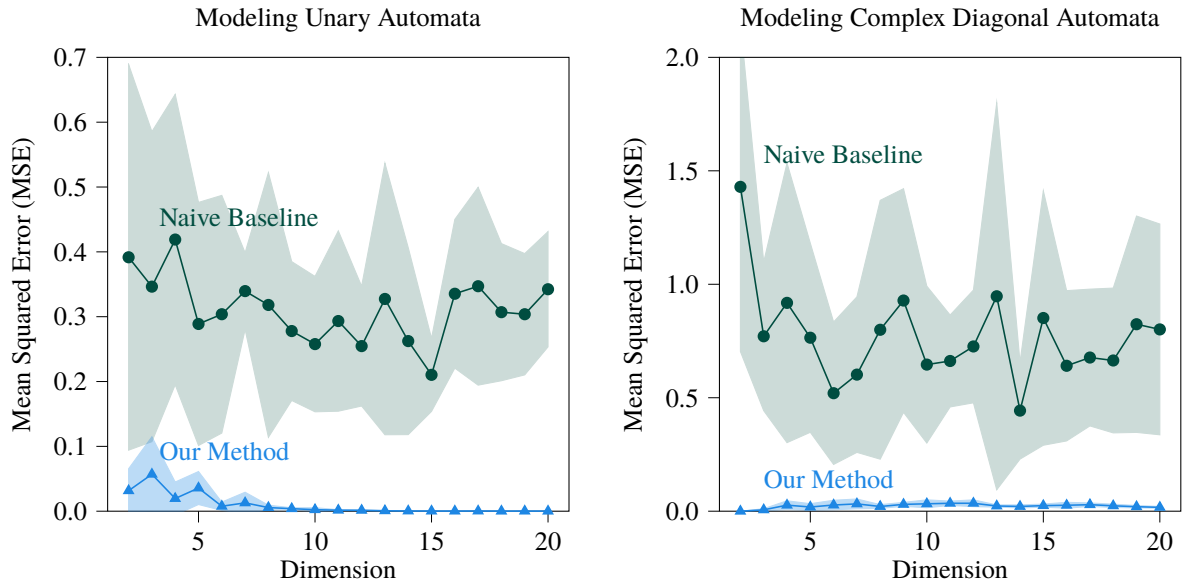
\begin{figure*}
    \centering
\begin{tikzpicture}
\definecolor{color0}{rgb}{0.84375,0.10546875,0.375}
\definecolor{color1}{rgb}{0.1171875,0.53125,0.89453125}
\definecolor{color2}{rgb}{0,0.30078125,0.25}
\definecolor{color3}{rgb}{0.99609375,0.75390625,0.02734375}

\begin{axis}[
name=ax1,
height=\figureheight,width=\figurewidth,
legend cell align={left},
legend style={fill opacity=0.8, draw opacity=1, text opacity=1, at={(1,0.5)}, anchor=west, draw=white!80!black},
tick align=outside,
tick pos=left,
title={Modeling Unary Automata},
x grid style={white!69.0196078431373!black},
xlabel={Dimension},
xmin=1.1, xmax=20.9,
xtick style={color=black},
y grid style={white!69.0196078431373!black},
ylabel={Mean Squared Error (MSE)}, ylabel near ticks,
ymin=0, ymax=0.7,
ytick style={color=black},
ytick={0,0.1,0.2,0.3,0.4,0.5,0.6,0.7},
yticklabels={0.0,0.1,0.2,0.3,0.4,0.5,0.6,0.7}
]
\path [draw=color2, fill=color2, opacity=0.2]
(axis cs:2,0.689196274990134)
--(axis cs:2,0.0941156866117145)
--(axis cs:3,0.107226127106975)
--(axis cs:4,0.194980105684975)
--(axis cs:5,0.101287691748954)
--(axis cs:6,0.120502734740395)
--(axis cs:7,0.280244224530877)
--(axis cs:8,0.113928187996778)
--(axis cs:9,0.170756085539777)
--(axis cs:10,0.153285866571921)
--(axis cs:11,0.154213297754367)
--(axis cs:12,0.162196464909031)
--(axis cs:13,0.117896204031415)
--(axis cs:14,0.118112781715981)
--(axis cs:15,0.15388702409852)
--(axis cs:16,0.220661815938779)
--(axis cs:17,0.194400665099929)
--(axis cs:18,0.201203121992104)
--(axis cs:19,0.210239874590211)
--(axis cs:20,0.253443054760677)
--(axis cs:20,0.431243017645524)
--(axis cs:20,0.431243017645524)
--(axis cs:19,0.397361472338127)
--(axis cs:18,0.412749765779577)
--(axis cs:17,0.499605750484937)
--(axis cs:16,0.450117602167033)
--(axis cs:15,0.266741655928342)
--(axis cs:14,0.406640442201712)
--(axis cs:13,0.536833594552753)
--(axis cs:12,0.347070097508781)
--(axis cs:11,0.432379139142632)
--(axis cs:10,0.36217160580767)
--(axis cs:9,0.38496183558875)
--(axis cs:8,0.522278984415962)
--(axis cs:7,0.398704059438178)
--(axis cs:6,0.487128798035955)
--(axis cs:5,0.476281229058792)
--(axis cs:4,0.642997644059863)
--(axis cs:3,0.585462991492162)
--(axis cs:2,0.689196274990134)
--cycle;

\path [draw=color1, fill=color1, opacity=0.3]
(axis cs:2,0.0654957741498947)
--(axis cs:2,-0.00228920951485634)
--(axis cs:3,-0.000836968421936035)
--(axis cs:4,-0.00592000037431717)
--(axis cs:5,0.0102813616394997)
--(axis cs:6,0.00120736099779606)
--(axis cs:7,-0.00355401169508696)
--(axis cs:8,0.00185155775398016)
--(axis cs:9,0.00248817424289882)
--(axis cs:10,-0.000772082712501287)
--(axis cs:11,0.000720475392881781)
--(axis cs:12,0.000561130116693676)
--(axis cs:13,0.000327853776980191)
--(axis cs:14,0.000146806036354974)
--(axis cs:15,9.83891222858801e-05)
--(axis cs:16,0.000122921526781283)
--(axis cs:17,0.00018162396736443)
--(axis cs:18,5.43491769349203e-05)
--(axis cs:19,3.63366125384346e-05)
--(axis cs:20,7.09749874658883e-05)
--(axis cs:20,0.000142012038850226)
--(axis cs:20,0.000142012038850226)
--(axis cs:19,0.000301490013953298)
--(axis cs:18,0.000184007556526922)
--(axis cs:17,0.000406881677918136)
--(axis cs:16,0.000254923768807203)
--(axis cs:15,0.000314371776767075)
--(axis cs:14,0.000561600318178535)
--(axis cs:13,0.00103755295276642)
--(axis cs:12,0.00218793051317334)
--(axis cs:11,0.00239184312522411)
--(axis cs:10,0.00545188039541245)
--(axis cs:9,0.00531170051544905)
--(axis cs:8,0.00932716298848391)
--(axis cs:7,0.0292817987501621)
--(axis cs:6,0.0140987066552043)
--(axis cs:5,0.0612265020608902)
--(axis cs:4,0.0450256168842316)
--(axis cs:3,0.114837124943733)
--(axis cs:2,0.0654957741498947)
--cycle;

\addplot [semithick, color2, mark=*, mark size=2, mark options={solid}]
table {%
2 0.391655980800924
3 0.346344559299569
4 0.418988874872419
5 0.288784460403873
6 0.303815766388175
7 0.339474141984527
8 0.31810358620637
9 0.277858960564263
10 0.257728736189795
11 0.293296218448499
12 0.254633281208906
13 0.327364899292084
14 0.262376611958847
15 0.210314340013431
16 0.335389709052906
17 0.347003207792433
18 0.30697644388584
19 0.303800673464169
20 0.3423430362031
}node[pos=0, inner sep=0.4cm, above right]{Naive Baseline};
\addlegendentry{Baseline}
\addplot [semithick, color1, mark=triangle*, mark size=2, mark options={solid}]
table {%
2 0.031603280454874
3 0.0570000782608986
4 0.0195528082549572
5 0.0357539318501949
6 0.00765303382650018
7 0.0128638939931989
8 0.00558936037123203
9 0.00389993749558926
10 0.00233989884145558
11 0.00155615922994912
12 0.00137453025672585
13 0.000682703335769475
14 0.000354203191818669
15 0.000206380456802435
16 0.000188922655070201
17 0.000294252822641283
18 0.000119178366730921
19 0.000168913320521824
20 0.000106493513158057
}node[pos=0, inner sep=0.4cm, above right]{Our Method};
\addlegendentry{Model MSE}
\legend{}
\end{axis}

\begin{axis}[
at={(ax1.south east)},
xshift=2cm,
height=\figureheight,width=\figurewidth,
legend cell align={left},
legend style={fill opacity=0.8, draw opacity=1, text opacity=1, at={(1,0.5)}, anchor=west, draw=white!80!black},
tick align=outside,
tick pos=left,
title={Modeling Complex Diagonal Automata},
x grid style={white!69.0196078431373!black},
xlabel={Dimension},
xmin=1.1, xmax=20.9,
xtick style={color=black},
y grid style={white!69.0196078431373!black},
ylabel={Mean Squared Error (MSE)}, ylabel near ticks,
ymin=0, ymax=2,
ytick style={color=black},
ytick={0, 0.5, 1, 1.5, 2},
yticklabels={0.0, 0.5, 1.0, 1.5, 2.0}
]
\path [draw=color2, fill=color2, opacity=0.2]
(axis cs:2,2.15495748461043)
--(axis cs:2,0.70362633898093)
--(axis cs:3,0.442937405894886)
--(axis cs:4,0.302018949664274)
--(axis cs:5,0.347651723089996)
--(axis cs:6,0.206037601501157)
--(axis cs:7,0.259177758466106)
--(axis cs:8,0.229201760228923)
--(axis cs:9,0.43555030477474)
--(axis cs:10,0.300261064443837)
--(axis cs:11,0.459340977282861)
--(axis cs:12,0.477503448113983)
--(axis cs:13,0.0941419871918839)
--(axis cs:14,0.229534410440114)
--(axis cs:15,0.290568977515248)
--(axis cs:16,0.309636348764392)
--(axis cs:17,0.375718926427968)
--(axis cs:18,0.345059707346479)
--(axis cs:19,0.347203954503055)
--(axis cs:20,0.336598113990905)
--(axis cs:20,1.26554730439874)
--(axis cs:20,1.26554730439874)
--(axis cs:19,1.30071213822515)
--(axis cs:18,0.983265378562644)
--(axis cs:17,0.978589488848119)
--(axis cs:16,0.97194689796159)
--(axis cs:15,1.41201459818325)
--(axis cs:14,0.656949842891228)
--(axis cs:13,1.80010376222985)
--(axis cs:12,0.974445766845401)
--(axis cs:11,0.864520684680045)
--(axis cs:10,0.991837636796176)
--(axis cs:9,1.42113761975402)
--(axis cs:8,1.36988073978674)
--(axis cs:7,0.944624187281792)
--(axis cs:6,0.833901126582115)
--(axis cs:5,1.18124695669371)
--(axis cs:4,1.53417006545451)
--(axis cs:3,1.09871495014053)
--(axis cs:2,2.15495748461043)
--cycle;

\path [draw=color1, fill=color1, opacity=0.3]
(axis cs:2,1.27993552534583e-05)
--(axis cs:2,-1.21820096193036e-06)
--(axis cs:3,0.00105290162764148)
--(axis cs:4,0.00741413484931536)
--(axis cs:5,0.00382889837018461)
--(axis cs:6,0.00615452262949409)
--(axis cs:7,0.0105412110164967)
--(axis cs:8,0.0124618597398587)
--(axis cs:9,0.0212795064642099)
--(axis cs:10,0.0151203784170632)
--(axis cs:11,0.0253393181604814)
--(axis cs:12,0.0199864694690435)
--(axis cs:13,0.018093444521187)
--(axis cs:14,0.0141966559534718)
--(axis cs:15,0.017782907470904)
--(axis cs:16,0.0130879217612641)
--(axis cs:17,0.0210313878249229)
--(axis cs:18,0.0164000057594983)
--(axis cs:19,0.0145303774571007)
--(axis cs:20,0.0131478970142144)
--(axis cs:20,0.0208714803529078)
--(axis cs:20,0.0208714803529078)
--(axis cs:19,0.0251062322168218)
--(axis cs:18,0.0311598731366966)
--(axis cs:17,0.0364321360688297)
--(axis cs:16,0.0392385744506785)
--(axis cs:15,0.0313984878577053)
--(axis cs:14,0.0277626832310188)
--(axis cs:13,0.0278624101600694)
--(axis cs:12,0.0501222139487623)
--(axis cs:11,0.0444113542977391)
--(axis cs:10,0.0502054611354246)
--(axis cs:9,0.0388226705953724)
--(axis cs:8,0.0290842763225345)
--(axis cs:7,0.0539262832260809)
--(axis cs:6,0.0482960261545697)
--(axis cs:5,0.0331783872532834)
--(axis cs:4,0.0459058560491812)
--(axis cs:3,0.0110718636834132)
--(axis cs:2,1.27993552534583e-05)
--cycle;

\addplot [semithick, color2, mark=*, mark size=2, mark options={solid}]
table {%
2 1.42929191179568
3 0.770826178017709
4 0.918094507559393
5 0.764449339891852
6 0.519969364041636
7 0.601900972873949
8 0.799541250007832
9 0.928343962264382
10 0.646049350620006
11 0.661930830981453
12 0.725974607479692
13 0.947122874710866
14 0.443242126665671
15 0.85129178784925
16 0.640791623362991
17 0.677154207638043
18 0.664162542954562
19 0.823958046364103
20 0.801072709194821
}node[pos=0, inner sep=0.3cm, above right]{Naive Baseline};
\addlegendentry{Baseline}
\addplot [semithick, color1, mark=triangle*, mark size=2, mark options={solid}]
table {%
2 5.79057714576396e-06
3 0.00606238265552735
4 0.0266599954492483
5 0.018503642811734
6 0.0272252743920319
7 0.0322337471212888
8 0.0207730680311966
9 0.0300510885297911
10 0.0326629197762439
11 0.0348753362291103
12 0.0350543417089029
13 0.0229779273406282
14 0.0209796695922453
15 0.0245906976643046
16 0.0261632481059713
17 0.0287317619468763
18 0.0237799394480975
19 0.0198183048369612
20 0.0170096886835611
}node[pos=0, inner sep=0.3cm, above right]{Our Method};
\addlegendentry{Model MSE}
\legend{}
\end{axis}

\end{tikzpicture}
    \caption{Results for Task 0: Training loss for modeling multiset automata with learned complex diagonal automata. For each set of data, the learned automaton with a complex diagonal transition matrix is able to approximate a unary (left) or diagonal (right) multiset automaton using the same number of states.  Error bands show $\pm 1$ standard deviation.}
    \label{fig:Unary_random}
\end{figure*}

\paragraph{Task 0: Recovering multiset automata}
To test how well complex diagonal automata can be trained from string weights, we generate a multiset automaton and train our model on strings together with their weights according to the automaton. Since we want to test the modeling power of complex multiset automata, we remove the final dense layer and replace it with a simple summation for this task only.  It is worth noting that this is equivalent to setting all final weights to $1$ by multiplying the final weights into the initial weights, therefore we lose no modeling power by doing this. The embedding dimension is set to match the number of states in the generated multiset automaton.  The size of the input alphabet is set to $1$ for the unary case and $5$ for the complex diagonal case.  We train by minimizing mean squared error.  
As a baseline for comparison, we compute the average string weight generated by each automaton and use that as the prediction for the weight of all strings generated by that automaton.

We generate unary automata by sampling uniformly from the Haar distribution over orthogonal matrices \cite{mezzadri:2007}.  The training strings are every unary string from length $0$ to $20$ ($21$ training strings).  Due to the small number of training strings, we let these models train for up to $30$k epochs with early stopping when loss does not improve for $100$ epochs. We generate complex diagonal automata by sampling real and imaginary coefficients uniformly from $[0,1]$ then renormalizing by the largest entry so the matrix has spectral radius $1$.  String lengths are fixed at $5$ to avoid large discrepancies in string weight magnitudes.  All strings of length $5$ are generated as training data. For each dimension, $10$ automata are generated.  We train $10$ models on each and select the best model as indicative of how well our model is capable of learning the original multiset automaton.

\begin{figure*}
    \centering
\begin{tikzpicture}

\definecolor{color0}{rgb}{0.83921568627451,0.152941176470588,0.156862745098039}
\definecolor{color1}{rgb}{1,0.498039215686275,0.0549019607843137}
\definecolor{color2}{rgb}{0,0.30078125,0.25}
\definecolor{color3}{rgb}{0.1171875,0.53125,0.89453125}

\begin{axis}[
name=ax1,
height=\figureheight,
legend cell align={left},
legend columns=-1,
legend style={font=\small, at={(0.5,-0.2)},anchor=north west},
tick align=outside,
tick pos=left,
title={Digit Sum},
width=\figurewidth,
x grid style={white!69.01960784313725!black},
xlabel={Number of digits},
xmin=5, xmax=95,
xtick style={color=black},
y grid style={white!69.01960784313725!black},
ylabel={Accuracy}, ylabel near ticks,
ymin=0, ymax=1.1,
ytick style={color=black},
ytick={0,0.2,0.4,0.6,0.8,1,1.2},
yticklabels={0.0,0.2,0.4,0.6,0.8,1.0,1.2}
]
\addplot [solid, thick, color0, mark=|, mark size=4, mark options={solid, rotate=45}]
table {%
5 1
10 1
15 1
20 1
25 1
30 1
35 1
40 1
45 1
50 1
55 1
60 1
65 1
70 1
75 1
80 1
85 1
90 1
95 1
};
\addlegendentry{DeepSets}
\addplot [solid, thick, color1, mark=|, mark size=4, mark options={solid}]
table {%
5 1
10 1
15 1
20 1
25 1
30 1
35 1
40 1
45 0.9995
50 0.9695
55 0.6865
60 0.2234
65 0.0256
70 0.0011
75 0.0001
80 0
85 0
90 0
95 0
};
\addlegendentry{LSTM}
\addplot [dashed, very thick, color2, mark=|, mark size=4, mark options={solid,rotate=-45}]
table {%
5 0.9981
10 0.9998
15 1
20 0.9999
25 1
30 1
35 0.9999
40 0.9995
45 0.9991
50 0.963
55 0.6751
60 0.2117
65 0.0257
70 0.0013
75 0.0001
80 0
85 0
90 0
95 0
};
\addlegendentry{GRU}
\addplot [dashed, thick, color3, mark=square*, mark size=1, mark options={solid}]
table {%
5 1
10 1
15 1
20 1
25 1
30 1
35 1
40 1
45 1
50 1
55 1
60 1
65 1
70 1
75 1
80 1
85 1
90 1
95 1
};
\addlegendentry{Our Method}
\end{axis}

\begin{axis}[
at={(ax1.south east)},
xshift=2cm,
height=\figureheight,
legend cell align={left},
legend style={at={(1,100)}, anchor=west, draw=white!80.0!black},
tick align=outside,
tick pos=left,
title={Digit Sum, Units Digit},
width=\figurewidth,
x grid style={white!69.01960784313725!black},
xlabel={Number of digits},
xmin=5, xmax=95,
xtick style={color=black},
y grid style={white!69.01960784313725!black},
ylabel={Accuracy}, ylabel near ticks,
ymin=0, ymax=1.1,
ytick style={color=black},
ytick={0,0.2,0.4,0.6,0.8,1,1.2},
yticklabels={0.0,0.2,0.4,0.6,0.8,1.0,1.2}
]
\addplot [solid, thick, color0, mark=|, mark size=4, mark options={solid, rotate=45}]
table {%
5 0.1
10 0.0991
15 0.0949
20 0.1032
25 0.1001
30 0.0977
35 0.0975
40 0.0978
45 0.1033
50 0.0997
55 0.0995
60 0.0996
65 0.1007
70 0.0973
75 0.1036
80 0.0986
85 0.0958
90 0.0975
95 0.0982
};
\addlegendentry{DeepSets}
\addplot [solid, thick, color1, mark=|, mark size=4, mark options={solid}]
table {%
5 0.0963
10 0.1003
15 0.0985
20 0.1034
25 0.103
30 0.0996
35 0.0989
40 0.0984
45 0.099
50 0.0976
55 0.1022
60 0.1053
65 0.1001
70 0.1002
75 0.099
80 0.1009
85 0.1011
90 0.096
95 0.0952
};
\addlegendentry{LSTM}
\addplot [dashed, very thick, color2, mark=|, mark size=4, mark options={solid,rotate=-45}]
table {%
5 0.0979
10 0.0987
15 0.0999
20 0.1034
25 0.098
30 0.0994
35 0.0983
40 0.1002
45 0.0977
50 0.0973
55 0.0979
60 0.104
65 0.1002
70 0.0985
75 0.0983
80 0.0987
85 0.0997
90 0.098
95 0.0932
};
\addlegendentry{GRU}
\addplot [dashed, thick, color3, mark=square*, mark size=1, mark options={solid}]
table {%
5 1
10 1
15 1
20 1
25 1
30 1
35 1
40 1
45 1
50 1
55 1
60 1
65 1
70 1
75 1
80 1
85 1
90 1
95 1
};
\addlegendentry{Our Method}
\legend{}
\end{axis}

\end{tikzpicture}
    \caption{Results for Task 1 (left) and Task 2 (right).  In task 1, the LSTM and GRU models were unable to generalize to examples larger than seen in training, while DeepSets and our model generalize to all test lengths.  For task 2, only our model is able to return the correct units digit for all test lengths.  The GRU, LSTM, and DeepSets models fail to learn any behavior beyond random guessing.}
    \label{fig:DigitSum}
\end{figure*}
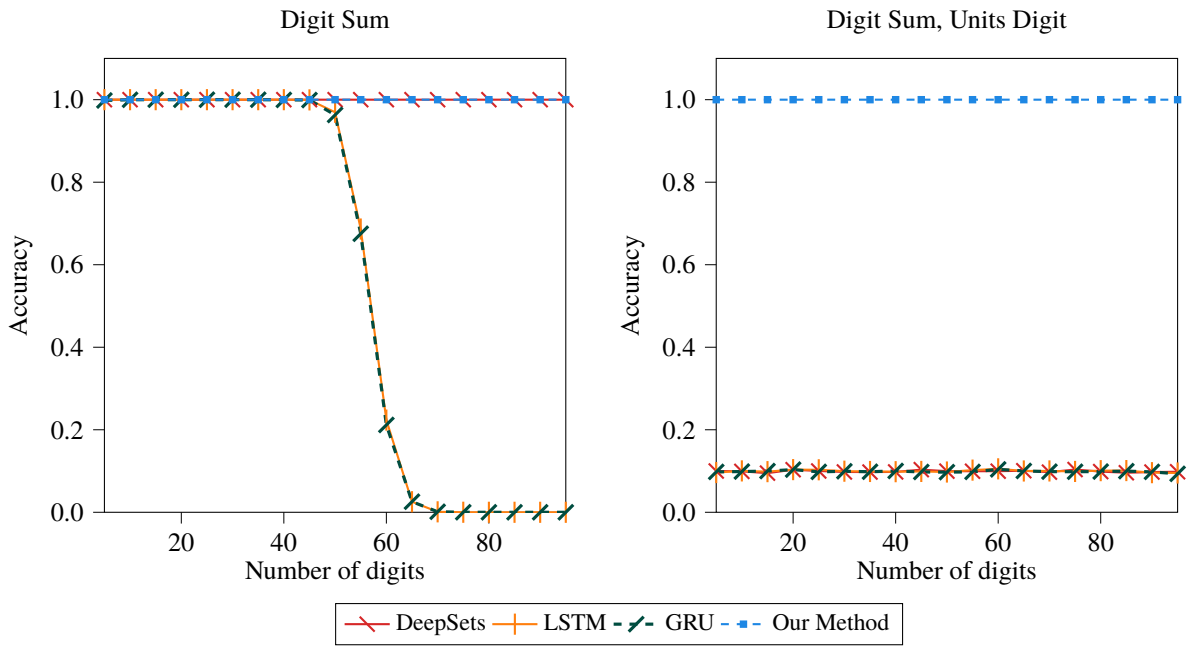

\paragraph{Task 1: Sum of digits}
In this task, taken from \citet{Zaheer2017DeepS}, the network receives a set of single digit integers as input and must output the sum of those digits.  The output is rounded to the nearest integer to measure accuracy. The training set consisted of 100k randomly generated sequences of digits 1--9 with lengths from 1 to 50.  They were fed to each network in the order in which they were generated (which only affects GRU and LSTM).  This was then split into training and dev with approximately a 99/1 split.  The test set consisted of randomly generated sequences of lengths that were multiples of 5 from 5 to 95. Figure~\ref{fig:DigitSum} shows that both our model and  DeepSets obtain perfect accuracy on the test data, while the LSTM and GRU fail to generalize to longer sequences.

\paragraph{Task 2: Returning units digit of a sum}
The second task is similar to the first, but only requires returning the units digit of the sum. The data and evaluation are otherwise the same as task 1. Here, random guessing within the output range of 0--9 achieves approximately $10\%$ accuracy.  Figure~\ref{fig:DigitSum} shows that DeepSets, LSTM, and GRU  are unable to achieve performance better than random guessing on the test data.  Our method is able to return the units digit perfectly for all test lengths, because it effectively learns to use the cyclic nature of complex multiplication to produce the units digit.   

\section{Conclusion}
We have proven that weighted multiset automata can be approximated by automata with (complex) diagonal transition matrices.  This formulation permits simpler elementwise multiplication instead of matrix multiplication, and requires fewer parameters when using the same number of states.  We show that this type of automaton naturally arises within existing neural architectures, and that this representation generalizes two existing multiset representations, the Transformer's position encodings and DeepSets. Our results provide new theoretical and intuitive justification for these models, and, in one case, lead to a change in the model that drastically improves its performance.

\section*{Acknowledgements}

We would like to thank Toan Nguyen for providing his implementation of the Transformer and answering many questions about it.

This research is based upon work supported in part by the Office of the Director of National Intelligence (ODNI), Intelligence Advanced Research Projects Activity (IARPA), via contract \#FA8650-17-C-9116. The views and conclusions contained herein are those of the authors and should not be interpreted as necessarily representing the official policies, either expressed or implied, of ODNI, IARPA, or the U.S. Government. The U.S. Government is authorized to reproduce and distribute reprints for governmental purposes notwithstanding any copyright annotation therein.

\bibliography{refs}

\begin{thebibliography}{17}
\providecommand{\natexlab}[1]{#1}
\providecommand{\url}[1]{\texttt{#1}}
\expandafter\ifx\csname urlstyle\endcsname\relax
  \providecommand{\doi}[1]{doi: #1}\else
  \providecommand{\doi}{doi: \begingroup \urlstyle{rm}\Url}\fi

\bibitem[DeBenedetto \& Chiang(2018)DeBenedetto and
  Chiang]{debenedetto+chiang:2018}
DeBenedetto, J. and Chiang, D.
\newblock Algorithms and training for weighted multiset automata and regular
  expressions.
\newblock In C{\^a}mpeanu, C. (ed.), \emph{Implementation and Application of
  Automata}, pp.\  146--158. Springer, 2018.

\bibitem[Devlin et~al.(2019)Devlin, Chang, Lee, and
  Toutanova]{devlin-etal-2019-bert}
Devlin, J., Chang, M.-W., Lee, K., and Toutanova, K.
\newblock {BERT}: Pre-training of deep bidirectional transformers for language
  understanding.
\newblock In \emph{Proc. NAACL HLT}, pp.\  4171--4186, 2019.

\bibitem[Gehring et~al.(2017)Gehring, Auli, Grangier, Yarats, and
  Dauphin]{gehring+:2017}
Gehring, J., Auli, M., Grangier, D., Yarats, D., and Dauphin, Y.~N.
\newblock Convolutional sequence to sequence learning.
\newblock In \emph{Proc. ICML}, 2017.

\bibitem[Horn \& Johnson(2012)Horn and Johnson]{horn+johnson:2012}
Horn, R.~A. and Johnson, C.~A.
\newblock \emph{Matrix Analysis}.
\newblock Cambridge Univ. Press, 2nd edition, 2012.

\bibitem[Kozma(2019)]{kozma:2019}
Kozma, L.
\newblock Useful inequalities, 2019.
\newblock URL \url{http://www.Lkozma.net/inequalities_cheat_sheet}.

\bibitem[Maron et~al.(2019)Maron, Ben-Hamu, Shamir, and
  Lipman]{maron2018invariant}
Maron, H., Ben-Hamu, H., Shamir, N., and Lipman, Y.
\newblock Invariant and equivariant graph networks.
\newblock In \emph{Proc. ICLR}, 2019.

\bibitem[Mezzadri(2007)]{mezzadri:2007}
Mezzadri, F.
\newblock How to generate random matrices from the classical compact groups.
\newblock \emph{Notices of the AMS}, 54\penalty0 (5):\penalty0 592--604, 2007.

\bibitem[Murphy et~al.(2019)Murphy, Srinivasan, Rao, and
  Ribeiro]{Murphy2018JanossyPL}
Murphy, R.~L., Srinivasan, B., Rao, V.~A., and Ribeiro, B.
\newblock {J}anossy pooling: Learning deep permutation-invariant functions for
  variable-size inputs.
\newblock In \emph{Proc. ICLR}, 2019.

\bibitem[O'Meara \& Vinsonhaler(2006)O'Meara and
  Vinsonhaler]{o2006approximately}
O'Meara, K. and Vinsonhaler, C.
\newblock On approximately simultaneously diagonalizable matrices.
\newblock \emph{Linear Algebra and its Applications}, 412\penalty0
  (1):\penalty0 39--74, 2006.

\bibitem[Pevn{\'y} \& Somol(2016)Pevn{\'y} and Somol]{Pevn2016UsingNN}
Pevn{\'y}, T. and Somol, P.
\newblock Using neural network formalism to solve multiple-instance problems.
\newblock In \emph{Proc. ISNN}, 2016.

\bibitem[Sakarovitch(2009)]{sakarovitch:2009}
Sakarovitch, J.
\newblock Rational and recognisable power series.
\newblock In \emph{Handbook of Weighted Automata}, pp.\  105--174. Springer,
  2009.

\bibitem[Saxe et~al.(2014)Saxe, McClelland, and Ganguli]{saxe+:2013}
Saxe, A.~M., McClelland, J.~L., and Ganguli, S.
\newblock Exact solutions to the nonlinear dynamics of learning in deep linear
  neural networks.
\newblock In \emph{Proc. ICLR}, 2014.

\bibitem[Vaswani et~al.(2017)Vaswani, Shazeer, Parmar, Uszkoreit, Jones, Gomez,
  Kaiser, and Polosukhin]{Vaswani2017AttentionIA}
Vaswani, A., Shazeer, N., Parmar, N., Uszkoreit, J., Jones, L., Gomez, A.~N.,
  Kaiser, L., and Polosukhin, I.
\newblock Attention is all you need.
\newblock In \emph{Proc. NeurIPS}, 2017.

\bibitem[Vinyals et~al.(2016)Vinyals, Bengio, and Kudlur]{vinyals2015order}
Vinyals, O., Bengio, S., and Kudlur, M.
\newblock Order matters: Sequence to sequence for sets.
\newblock In \emph{Proc. ICLR}, 2016.

\bibitem[Wagstaff et~al.(2019)Wagstaff, Fuchs, Engelcke, Posner, and
  Osborne]{Wagstaff2019OnTL}
Wagstaff, E., Fuchs, F.~B., Engelcke, M., Posner, I., and Osborne, M.~A.
\newblock On the limitations of representing functions on sets.
\newblock In \emph{Proc. ICML}, 2019.

\bibitem[Yang et~al.(2020)Yang, Wang, Markham, and Trigoni]{yang2019}
Yang, B., Wang, S., Markham, A., and Trigoni, N.
\newblock Robust attentional aggregation of deep feature sets for multi-view
  {3D} reconstruction.
\newblock \emph{International Journal of Computer Vision}, 128:\penalty0
  53--73, 2020.

\bibitem[Zaheer et~al.(2017)Zaheer, Kottur, Ravanbakhsh, P{\'o}czos,
  Salakhutdinov, and Smola]{Zaheer2017DeepS}
Zaheer, M., Kottur, S., Ravanbakhsh, S., P{\'o}czos, B., Salakhutdinov, R., and
  Smola, A.~J.
\newblock Deep sets.
\newblock In \emph{Proc. NeurIPS}, 2017.

\end{thebibliography}

\clearpage

\appendix

\section{Proof of Theorem~\ref{prop:bound}}
\label{sec:bound}
Our strategy is to form a Jordan decomposition of $A$ and show that the desired bounds hold for each Jordan block. To this end, we first prove the following lemmas.

\begin{lemma} \label{lem:nonzero}
If $J$ is a Jordan block with nonzero eigenvalue, then for any $\epsilon > 0$ there is a complex matrix $D$ such that $J+D$ is diagonalizable in $\mathbb{C}$ and \[ \frac{\|(J+D)^n-J^n\|}{\|J^n\|} \leq n\epsilon. \]
\end{lemma}

\begin{proof}
The powers of $J$ look like
\begin{align*}
J^n &= 
\begin{bmatrix}
\binom{n}{0} \lambda^n & \binom{n}{1} \lambda^{n-1} & \binom{n}{2} \lambda^{n-2} & \cdots \\
& \binom{n}{0} \lambda^n & \binom{n}{1} \lambda^{n-1} & \cdots \\
& & \binom{n}{0} \lambda^n & \cdots \\
& & & \ddots
\end{bmatrix}.
\end{align*}
More concisely,
\[
[J^n]_{jk} =
\begin{cases}
\binom{n}{k-j} \, \lambda^{n-k+j} & \text{if $0 \leq k - j \leq n$} \\
0 & \text{otherwise.}
\end{cases}\]

We choose $D$ to perturb the diagonal elements of $J$ towards zero; that is, let $D$ be a diagonal matrix whose elements are in $[-\epsilon\lambda, 0)$ and are all different. 
This shrinks the diagonal elements by a factor no smaller than $(1-\epsilon)$. So the powers of $(J+D)$ are, for $0 \leq k-j \leq n$:
\begin{align*}
[(J+D)^n]_{jk} &= c_{jk} \, [J^n]_{jk} \\
c_{jk} &\geq (1-\epsilon)^{n-k+j}.
\end{align*}
Simplifying the bound on $c_{jk}$ \citep{kozma:2019}:
\begin{equation}
c_{jk} \geq 1-(n-k+j)\epsilon  \geq 1-n\epsilon.
\end{equation}
The elements of $J^n$, for $0 \leq k-j \leq n$, are perturbed by:
\begin{align*}
[(J+D)^n-J^n]_{jk} &= (c_{jk}-1) [J^n]_{jk} \\
\left|[(J+D)^n-J^n]_{jk}\right| & \leq n\epsilon \left|[J^n]_{jk} \right|. \\
\intertext{Since $\|\mathord\cdot\|$ is monotonic,}
\left\|(J+D)^n-J^n\right\| & \leq n\epsilon \left\|J^n\right\| \\
\frac{\left\|(J+D)^n-J^n\right\|}{\left\|J^n\right\|} & \leq n\epsilon. \qedhere
\end{align*}

\end{proof}

\begin{lemma} \label{lem:zero}
If $J$ is a Jordan block with zero eigenvalue, then for any $\epsilon>0, r>0$, there is a complex matrix $D$ such that $J+D$ is diagonalizable in $\mathbb{C}$ and
\[ \|(J+D)^n - J^n \| \leq r^n\epsilon.\] \end{lemma}

\begin{proof}
Since the diagonal elements of $J$ are all zero, we can't perturb them toward zero as in Lemma~\ref{lem:nonzero}; instead, let
\begin{align*}
\delta &= \min \left\{ \frac{r}2, \left(\frac{r}{2}\right)^d \frac{\epsilon}{d} \right\}
\end{align*}
and let $D$ be a diagonal matrix whose elements are in $(0, \delta]$ and are all different.
Then the elements of $((J+D)^n-J^n)$ are, for $0 \leq k-j < \min\{n, d\}$:
\begin{align*}
[(J+D)^n-J^n]_{jk} 
&\leq \binom{n}{k-j} \, \delta^{n-k+j} \\
&< 2^n \delta^{n-k+j} \\
&\leq 2^n \delta^{\min\{0,n-d\}+1}, \\
\intertext{and by monotonicity,}
\|(J+D)^n-J^n\| &\leq 2^n \delta^{\min\{0,n-d\}+1} d.
\end{align*}
To simplify this bound, we consider two cases. If $n \leq d$,
\begin{align*}
\|(J+D)^n-J^n\| &= 2^n \delta d \\ &\leq 2^n \left(\frac{r}2\right)^d \frac\epsilon{d} d \\ &= 2^{n-d} r^d \epsilon \\ &\leq r^n \epsilon.
\end{align*}
If $n > d$,
\begin{align*}
\|(J+D)^n-J^n\| &= 2^n \delta^{n-d+1} d \\
&\leq 2^n \delta^{n-d} \left(\frac{r}2\right)^d \frac\epsilon{d} d \\
&\leq 2^n \left(\frac{r}2\right)^{n-d} \left(\frac{r}2\right)^d \frac\epsilon{d} d \\
&= r^n \epsilon. \qedhere
\end{align*}
\end{proof}

Now we can combine the above two lemmas to obtain the desired bounds for a general matrix.
\begin{proof}[Proof of Theorem~\ref{prop:bound}]
Form the Jordan decomposition $A = PJP^{-1}$, where
\begin{align*}
J &= \begin{bmatrix}
J_1 & \\
& J_2 & \\
& & \ddots \\
& & & J_p \\
\end{bmatrix}
\end{align*}
and each $J_j$ is a Jordan block. 
Let $\kappa(P) = \|P\| \|P^{-1}\|$ be the Frobenius condition number of $P$. 

If $A$ is nilpotent, use Lemma~\ref{lem:zero} on each block $J_j$ to find a $D_j$ so that $\|(J_j + D_j)^n-J_j^n\| \leq \frac{r^n \epsilon}{\kappa(P) p}$. Combine the $D_j$ into a single matrix $D$, so that $\|(J + D)^n-J^n\| \leq \frac{r^n\epsilon}{\kappa(P)}$. 
Let $E = PDP^{-1}$, and then
\begin{align*}
\|(A + E)^n - A^n\| &= \|P((J+D)^n-J^n)P^{-1}\| \\
&\leq \kappa(P)\|(J+D)^n-J^n\| \\ 
&\leq \kappa(P) \frac{r^n \epsilon}{\kappa(P)} \\
&= r^n \epsilon.
\end{align*}

If $A$ is not nilpotent, then for each Jordan block $J_j$:
\begin{itemize}
\item If $J_j$ has nonzero eigenvalue, use Lemma~\ref{lem:nonzero} to find a $D_j$ such that  $\|(J_j + D_j)^n-J_j^n\| \leq \frac{n\epsilon}{\kappa(P)^2} \frac{\| J^n \|}{2p}$.
\item If $J_j$ has zero eigenvalue, use Lemma~\ref{lem:zero} to find a $D_j$ such that $\|(J_j + D_j)^n-J_j^n\| \leq \frac{n\epsilon}{\kappa(P)^2} \frac{\rho(J)^n}{2p} $.
\end{itemize}
Combine the $D_j$ into a single matrix $D$.
Then the total absolute error of all the blocks with nonzero eigenvalue is at most $\frac{n\epsilon}{\kappa(P)^2} \frac{\| J^n \|}{2}$.
And since $\rho(J)^n \leq \|J^n\|$, the total absolute error of all the blocks with zero eigenvalue is also at most $\frac{n\epsilon}{\kappa(P)^2} \frac{ \|J^n\|}{2}$. So the combined total is
\[\|(J+D)^n - J^n\| \leq \frac{n\epsilon}{\kappa(P)^2}  \|J^n\|.\]

Finally, let $E = PDP^{-1}$, and
\begin{align*}
\|(A+E)^n - A^n\| &= \|P((J+D)^n-J^n)P^{-1}\| \\
&\leq \kappa(P) \|((J+D)^n-J^n)\| \\
&\leq  \frac{n\epsilon}{\kappa(P)} \|J^n\| \\
&\leq \frac{n\epsilon}{\kappa(P)} \|P^{-1}A^nP\| \\
&\leq n\epsilon \|A^n\| \\
\frac{\|(A+E)^n - A^n\|}{\|A^n\|} &\leq n\epsilon. \qedhere
\end{align*}
\end{proof}

\section{Proof of Proposition~\ref{prop:regops}}
\label{sec:regops}
First, consider the $\oplus$ operation.
Let $\mu_1(a)$ (for all $a$) be the transition matrices of $M_1$. For any $\epsilon>0$, let $E_1(a)$ be the perturbations of the $\mu_1(a)$ such that $\|E_1(a)\| \leq \epsilon/2$ and the $\mu_1(a) + E_1(a)$ (for all $a$) are simultaneously diagonalizable. Similarly for $M_2$. Then the matrices $(\mu_1(a) + E_1(a)) \oplus (\mu_2(a) + E_2(a))$ (for all $a$) are simultaneously diagonalizable, and 
\begin{equation*}
\begin{split}
&\| (\mu_1(a) + E_1(a)) \oplus (\mu_2(a) + E_2(a)) - \mu_1(a) \oplus \mu_2(a) \| \\&\quad = \| E_1(a) \oplus E_2(a) \| \\ &\quad \leq \|E_1(a)\| + \|E_2(a)\| \\&\quad \leq \epsilon.
\end{split}
\end{equation*}

Next, we consider the $\shuffle$ operation.
Let $d_1$ and $d_2$ be the number of states in $M_1$ and $M_2$, respectively. Let $E_1(a)$ be the perturbations of the $\mu_1(a)$ such that $\|E_1(a)\| \leq \epsilon/(2d_2)$ and the $\mu_1(a)+E_1(a)$ are simultaneously diagonalizable by some matrix $P_1$. Similarly for $M_2$. 

Then the matrices $(\mu_1(a) + E_1(a)) \shuffle (\mu_2(a) + E_2(a))$ (for all $a$) are simultaneously diagonalizable by $P_1 \otimes P_2$. 
To see why, let $A_1 = \mu_1(a) + E_1(a)$ and $A_2 = \mu_2(a)+E_2(a)$ and
observe that
\begin{equation*}
\begin{split}
&(P_1 \otimes P_2)(A_1 \shuffle A_2)(P_1 \otimes P_2)^{-1} \\ 
&\quad = (P_1 \otimes P_2)(A_1 \otimes I + I \otimes A_2)(P_1^{-1} \otimes P_2^{-1}) \\
&\quad = P_1 A_1 P_1^{-1} \otimes I + I \otimes P_2 A_2 P_2^{-1} \\
&\quad = P_1A_1P_1^{-1} \shuffle P_2A_2P_2^{-1},
\end{split}
\end{equation*}
which is diagonal.

To show that $(\mu_1(a) + E_1(a)) \shuffle (\mu_2(a) + E_2(a))$ is close to $(\mu_1(a) \shuffle \mu_2(a)$, observe that
\begin{equation*}
\begin{split}
&(\mu_1(a) + E_1(a)) \shuffle (\mu_2(a) + E_2(a)) \\
&\quad = (\mu_1(a) + E_1(a)) \otimes I + I \otimes (\mu_2(a) + E_2(a)) \\
&\quad = \mu_1(a) \otimes I + E_1(a) \otimes I + I \otimes \mu_2(a) + I \otimes E_2(a) \\
&\quad = (\mu_1(a) \shuffle \mu_2(a)) + (E_1(a) \shuffle E_2(a)).
\end{split}
\end{equation*}
Therefore,
\begin{equation*}
\begin{split}
&\| (\mu_1(a) + E_1(a)) \shuffle (\mu_2(a) + E_2(a)) - \mu_1(a) \shuffle \mu_2(a) \| \\
&\quad = \| E_1(a) \shuffle E_2(a) \| \\
&\quad = \| E_1(a) \otimes I + I \otimes E_2(a) \| \\
&\quad \leq \| E_1(a) \otimes I \| + \| I \otimes E_2(a) \| \\
&\quad \leq \|E_1(a)\| d_2 + d_1 \|E_2(a)\| \\
&\quad \leq \epsilon.
\end{split}
\end{equation*}

\section{Proof of Proposition~\ref{prop:makeasd}}
\label{sec:makeasd}
Because any set of commuting matrices can be simultaneously triangularized by a change of basis, assume without loss of generality that $M$'s transition matrices are upper triangular, that is, there are no transitions from state~$q$ to state $r$ where $q>r$.

Let $M = (Q, \Sigma, \lambda, \mu, \rho)$, and arbitrarily number the symbols of $\Sigma$ as $a_1, \ldots, a_m$.
Note that $M$ assigns the same weight to multiset $w$ as it does to the sorted symbols of $w$. That is, we can compute the weight of $w$ by summing over sequences of states $q_0, \ldots, q_m$ such that $q_0$ is an initial state, $q_m$ is a final state, and $M$ can get from state $q_{i-1}$ to $q_{i}$ while reading $a_i^k$, where $k$ is the number of occurrences of $a_i$ in $w$. 

For all $a \in \Sigma, q, r \in Q$, define $M_{q,a,r}$ to be the automaton that assigns to $a^k$ the same weight that $M$ would going from state $q$ to state $r$ while reading $a^k$. That is,
\begin{align*}
M_{q,a,r} &= (\lambda_{q,a,r}, \mu_{q,a,r}, \rho_{q,a,r}) \\
[\lambda_{q,a,r}]_q &= 1 \\
\mu_{q,a,r}(a) &= \mu(a) \\
[\rho_{q,a,r}]_r &= 1
\end{align*}
and all other weights are zero.

Then we can build a multiset automaton equivalent to $M$ by combining the $M_{q,a,r}$ using the union and shuffle operations:
\begin{equation*}
M' = \bigoplus_{\substack{q_0, \ldots, q_m \in Q\\q_0 \leq \cdots \leq q_m}}  \lambda_{q_0} M_{q_0,a_1,q_1} \shuffle \cdots \shuffle M_{q_{m-1},a_m,q_m} \rho_{q_m}
\end{equation*}
(where multiplying an automaton by a scalar means scaling its initial or final weight vector by that scalar). The $M_{q,a,r}$ are unary, so by Proposition~\ref{prop:regops}, the transition matrices of $M'$ are ASD.
Since $M_{q,a,r}$ has $r-q+1$ states, the number of states in $M'$ is 
\begin{equation*}
|Q'| = \sum_{q_0 \leq \cdots \leq q_m} \prod_{i=1}^{m} (q_i-q_{i-1}+1)
\end{equation*}
which we can find a closed-form expression for using generating functions. If $p(z)$ is a polynomial, let $[z^{i}](p(z))$ stand for ``the coefficient of $z^{i}$ in $p$.'' Then
\begin{align*}
|Q'| 
&= \left[z^{d-1}\right] \left(\sum_{i=0}^{\infty} z^i\right) \left(\sum_{i=0}^{\infty} (i+1)z^i\right)^m \left(\sum_{i=0}^{\infty} z^i\right) \\
&= \left[z^{d-1}\right] \left(\frac1{1-z}\right) \left(\frac1{1-z}\right)^{2m} \left(\frac1{1-z}\right)\\
&= \left[z^{d-1}\right] \left(\frac1{1-z}\right)^{2m+2} \\
&= \binom{2m+d}{d-1}.
\end{align*}

\end{document}